\RequirePackage[loading]{tracefnt}
\documentclass{article}

\usepackage[final]{neurips_2021}  

\usepackage[utf8]{inputenc} 
\usepackage[T1]{fontenc}    
\usepackage{hyperref}       
\usepackage{url}            
\usepackage{booktabs}       
\usepackage{amsfonts}       
\usepackage{nicefrac}       
\usepackage{microtype}      
\usepackage{xcolor}         
\usepackage{makecell}
\usepackage[english]{babel}
\usepackage{caption}
\usepackage{multirow,array}
\usepackage[T1]{fontenc}
\usepackage[mathscr]{euscript}
\usepackage{float}
\usepackage{tikz}
\usepackage{tabularx}
\usetikzlibrary{arrows}
\usepackage[verbose]{wrapfig}
\usepackage{picins}

\usepackage{pifont}
\newcommand{\cmark}{\ding{51}}%
\newcommand{\xmark}{\ding{55}}%

\usepackage{textcomp}

\usepackage{subfig}
\usepackage{xparse,hhline,array,multirow}
\newcommand\doubleactivetilde{~~} 

\ExplSyntaxOn

\seq_new:N \l_cfr_game_body_seq
\seq_new:N \l_cfr_game_first_seq
\tl_new:N \l_cfr_game_first_tl
\dim_new:N \l__cfr_game_dim
\cs_generate_variant:Nn \seq_set_split:Nnn { NnV }

\cs_new_protected:Nn \__cfr_game_first:n
 {
  & \multicolumn{1}{c}{#1}
 }

\NewDocumentEnvironment{game}{mmooob}
 {

  \seq_set_split:Nnn \l_cfr_game_body_seq { \\ } { #6 }
  \seq_pop_right:NN \l_cfr_game_body_seq \l_tmpa_tl
  \tl_if_empty:NF \l_tmpa_tl { \seq_put_right:NV \l_cfr_game_body_seq \l_tmpa_tl }
  \seq_pop_left:NN \l_cfr_game_body_seq \l_cfr_game_first_tl
  \seq_set_split:NnV \l_cfr_game_first_seq { & } \l_cfr_game_first_tl
  \seq_pop_left:NN \l_cfr_game_first_seq \l_tmpa_tl
  \__cfr_game_equalize_columns:
  \begin{tabular}{@{}c@{}r|*{#2}{w{c}{\l__cfr_game_dim}|}}
  \IfValueT{#4}{\multicolumn{2}{@{}c}{} & \multicolumn{#2}{c}{#4} \\}
  \multicolumn{2}{c}{} \seq_map_function:NN \l_cfr_game_first_seq \__cfr_game_first:n \\
  \exp_args:No \hhline{\doubleactivetilde*{#2}{|-}|}
  \IfValueTF{#3}{\multirow{#1}{*}{#3\quad}}{\multicolumn{1}{@{}c@{}}{}}
  & \seq_use:Nn \l_cfr_game_body_seq
   { \\ \exp_args:No \hhline{\doubleactivetilde*{#2}{|-}|} & } \\
  \exp_args:No \hhline{\doubleactivetilde*{#2}{|-}|}
  \IfValueT{#5}
   {
    \multicolumn{2}{c}{\rule{0pt}{1.25\normalbaselineskip}} &
    \multicolumn{#2}{c}{\makebox[0pt]{#5}} \\
   }
  \end{tabular}
}{}

\cs_new_protected:Nn \__cfr_game_equalize_columns:
 {
  \dim_zero:N \l__cfr_game_dim
  \seq_map_function:NN \l_cfr_game_first_seq \__cfr_game_measure_cell:n
  \seq_map_inline:Nn \l_cfr_game_body_seq
   {
    \__cfr_game_measure_columns:n { ##1 }
   }
 }
\cs_new_protected:Nn \__cfr_game_measure_columns:n
 {
  \seq_set_split:Nnn \l_tmpa_seq { & } { #1 }
  \seq_indexed_map_inline:Nn \l_tmpa_seq
   {
    \int_compare:nT { ##1 > 1 }
     {
      \__cfr_game_measure_cell:n { ##2 }
     }
   }
 }
\cs_new_protected:Nn \__cfr_game_measure_cell:n
 {
  \hbox_set:Nn \l_tmpa_box { #1 }
  \dim_set:Nn \l__cfr_game_dim { \dim_max:nn { \l__cfr_game_dim } { \box_wd:N \l_tmpa_box } }
 }
 
\NewDocumentEnvironment{game_smaller}{mmooob}
 {

  \seq_set_split:Nnn \l_cfr_game_body_seq { \\ } { #6 }
  \seq_pop_right:NN \l_cfr_game_body_seq \l_tmpa_tl
  \tl_if_empty:NF \l_tmpa_tl { \seq_put_right:NV \l_cfr_game_body_seq \l_tmpa_tl }
  \seq_pop_left:NN \l_cfr_game_body_seq \l_cfr_game_first_tl
  \seq_set_split:NnV \l_cfr_game_first_seq { & } \l_cfr_game_first_tl
  \seq_pop_left:NN \l_cfr_game_first_seq \l_tmpa_tl
  \begin{tabular}{@{}c@{}r*{#2}{|c}|}
  \IfValueT{#4}{\multicolumn{2}{@{}c}{} & \multicolumn{#2}{c}{#4} \\}
  \multicolumn{2}{c}{} \seq_map_function:NN \l_cfr_game_first_seq \__cfr_game_first:n \\
  \exp_args:No \hhline{\doubleactivetilde*{#2}{|-}|}
  \IfValueTF{#3}{\multirow{#1}{*}{#3\quad}}{\multicolumn{1}{@{}c@{}}{}}
  & \seq_use:Nn \l_cfr_game_body_seq
   { \\ \exp_args:No \hhline{\doubleactivetilde*{#2}{|-}|} & } \\
  \exp_args:No \hhline{\doubleactivetilde*{#2}{|-}|}
  \IfValueT{#5}
   {
    \multicolumn{2}{c}{\rule{0pt}{1.25\normalbaselineskip}} &
    \multicolumn{#2}{c}{\makebox[0pt]{#5}} \\
   }
  \end{tabular}
}{}
\ExplSyntaxOff

\usepackage{csquotes}
\usepackage[normalem]{ulem}
\usepackage{amsmath}
\allowdisplaybreaks
\usepackage{amssymb}
\usepackage{amsthm}
\usepackage{amstext}
\usepackage{mathtools}
\usepackage{amsfonts}
\usepackage[bb=dsserif]{mathalpha}
\usepackage{bm}
\usepackage{bbm}
\usepackage{graphicx}
\usepackage{tikz,}
\usetikzlibrary{patterns}
\usepackage[ruled,noend]{algorithm2e}

\SetCommentSty{mycommfont}
\usepackage[noend]{algpseudocode}
\usepackage[stable]{footmisc}
\usepackage{setspace}
\usepackage[mathscr]{euscript}
\usepackage{comment}

\theoremstyle{definition}
\newtheorem{definition}{Definition}[section]
\newcommand{\Tra}{^{\sf T}} 

\newcommand{\Act}{{\mathcal{A}}}

\newcommand{\States}{\mathcal{S}}

\newcommand{\expect}{\mathbb{E}}

\newcommand{\amTFT}{\texttt{amTFT}}

\newcommand{\Histories}{\mathcal{H}}

\newcommand{\welfareUtil}{w^{\mathrm{Util}}}
\newcommand{\welfareIneq}{w^{\mathrm{IA}}}

\newcommand{\expit}{\mathrm{expit}}

\newcommand{\mi}{_{-i}}

\newcommand{\Pihat}{\Pi}

\newcommand{\WelfareFunctions}{\mathcal{W}}
\newcommand{\policyProfile}{\pi}

\newcommand{\ii}{^{i}_i}
\newcommand{\imi}{^{i}_{-i}}
\newcommand{\Vbar}{\overline{V}}
\newcommand{\Vexploit}{V^{\mathrm{exploitation}}}
\newcommand{\Vrobust}{V^{\mathrm{robustness}}}

\newcommand{\Norms}{\mathcal{N}}
\DeclareMathOperator*{\argmax}{arg\,max}
\DeclareMathOperator*{\argmin}{arg\,min}

\theoremstyle{plain}

\newtheorem{prop}{Proposition}
\title{Normative Disagreement as a Challenge for Cooperative AI}

\author{
Julian Stastny \\
University of Tuebingen\\
\texttt{julianstastny@gmail.com} \\

\And Maxime Rich\'{e} \\
Center on Long-Term Risk\\
\texttt{maxime.riche@longtermrisk.org} \\

\And {Alexander Lyzhov}\thanks{Now at New York University} \\
Center on Long-Term Risk \\
\texttt{alex.grig.lyzhov@gmail.com} \\

\And Johannes Treutlein \\
Vector Institute\\
University of Toronto\\
\texttt{treutlein@cs.toronto.edu} \\

\And Allan Dafoe \\
DeepMind\\
\texttt{allandafoe@deepmind.com}\\

\And Jesse Clifton \\
Center on Long-Term Risk\\
\texttt{jesse.clifton@longtermrisk.org} \\
}

\begin{document}

\maketitle

\begin{abstract} 
Cooperation in settings where agents have both common and conflicting interests (mixed-motive environments) has recently received considerable attention in multi-agent learning. However, the mixed-motive environments typically studied have a single cooperative outcome on which all agents can agree. Many real-world multi-agent environments are instead bargaining problems (BPs): they have several Pareto-optimal payoff profiles over which agents have conflicting preferences. 
We argue that typical cooperation-inducing learning algorithms fail to cooperate in BPs when there is room for \textit{normative disagreement} resulting in the existence of multiple competing cooperative equilibria, and illustrate this problem empirically.
To remedy the issue, we introduce the notion of \textit{norm-adaptive} policies. Norm-adaptive
policies are capable of behaving according to different norms in different circumstances, 
creating opportunities for resolving normative disagreement.  
We develop a class of norm-adaptive policies and show in experiments that these
significantly increase cooperation. 
However, norm-adaptiveness cannot address residual bargaining failure arising from a fundamental tradeoff between 
exploitability and cooperative robustness.  
\end{abstract}

\section{Introduction}
\label{sec:intro}

Multi-agent contexts often exhibit opportunities for cooperation: situations where joint action can lead to mutual benefits \citep{dafoe2020open}.
Individuals can engage in mutually beneficial trade; nation-states can enter into treaties instead of going to war; disputants can settle out of court rather than engaging in costly litigation. 
But a hurdle common to each of these examples is that
the agents will disagree about their ideal agreement.
Even if agreements benefit all parties relative to the status 
quo, different agreements will benefit different parties to different 
degrees. These circumstances can be called \textit{bargaining problems} \citep{schelling1956essay}.

\par As AI systems are deployed to act on behalf of humans in more real-world circumstances, they will need to be able 
to act effectively in 
bargaining problems — from commercial negotiations in the 
nearer-term (e.g., \citet{chakraborty2020automated}) to high-stakes strategic decision-making in 
the longer-term 
\citep{geist2018might}. 
Moreover, 
agents may be trained independently and offline before 
interacting with one another in the world. This raises concerns about
future AI systems following incompatible norms for arriving at solutions to 
bargaining problems, analogously to disagreements about fairness which 
create hurdles to international cooperation on critical issues 
such as climate policy \citep{albin2001justice,ringius2002burden}. 

\par
Our contributions 
are as follows. We
introduce a taxonomy of cooperation games, including
bargaining problems
(Section \ref{sec:taxonomy}). 
We relate their difficulty 
to the degree of 
\textit{normative disagreement}, i.e., 
differences over principles for selecting
among mutually beneficial outcomes, 
which we formalize in terms of 
\textit{welfare functions}.
Normative disagreement does
not arise in purely cooperative games or simple 
sequential social dilemmas \citep{leibo2017multi}, 
but
is an important obstacle for cooperation in what we call \textit{asymmetric} 
bargaining problems. Following this, we introduce the notion of 
\textit{norm-adaptive} policies -- policies which can play according 
to different norms depending on the circumstances. 
In several multi-agent learning environments we highlight the difficulty of 
bargaining between norm-unadaptive policies
(Section \ref{sec:simple}).  
We then contrast this with
a class of norm-adaptive policies
(Section \ref{sec:robustness-training}) based
on \citet{lerer2017maintaining}'s approximate Markov
tit-for-tat algorithm. 
We show that 
this
improves performance
in 
bargaining problems. 
However, there remain limitations, 
most fundamentally a tradeoff between exploitability and the robustness of cooperation. 

\section{Related work}

\begin{wrapfigure}{r}{0.5\textwidth}
    \vspace{-15mm}

    \begin{tikzpicture}[scale=0.6, every node/.style={scale=0.6}]
    	\begin{scope} [fill opacity = .2]
        \draw[fill=blue, draw = black] (-1.3, 0) circle (4);
        \draw[fill=orange, draw = black] (1.3, 0) circle (4);
        \draw[fill=blue, draw = black] (-1.8 ,-1.7) circle (2.2);
        \draw[fill=red, draw = black] (0, 0.0) circle (2.5);
        \draw[fill=red, draw = black] (0, 0.0) circle (1.5);
    
        \node at (0, 1.9) [align=center, opacity=1]{\large Bargaining};
        \node at (0, 0) [align=center, opacity=1]{Asymmetric \\ Bargaining};
        \node at (-4, 1) [align=center, opacity=1]{\large Mixed- \\ \large motive};
        \node at (4, 0) [align=center, opacity=1]{\large Coordination};
        \node at (-3, -2) [align=center, opacity=1]{\large SSDs};
    
        \end{scope}
        
    
    \end{tikzpicture}
    \caption{Venn diagram of cooperation problems. 
    }
    \label{fig:venn-diagram}
    \vspace{-5mm}

\end{wrapfigure}
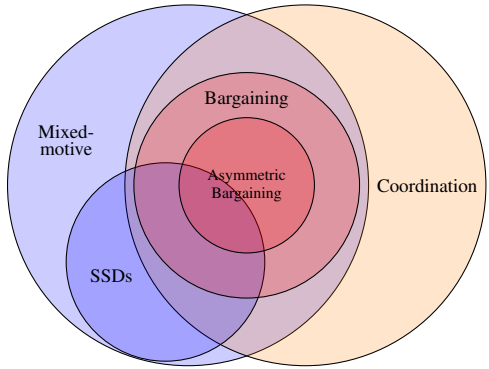

The field of multi-agent learning (MAL) has
recently
paid considerable attention to problems of
cooperation in mixed-motive games, in which 
agents have conflicting preferences.
Much of
this work has been focused on
\textit{sequential social dilemmas (SSDs)} (e.g., 
\citealt{peysakhovich2017consequentialist,lerer2017maintaining,eccles2019learning}).
The classic example of a social dilemma
is the Prisoner's Dilemma, and the SSDs studied
in this literature are similar to the Prisoner's Dilemma
in that there is a single salient notion of ``cooperation''. 
This means that it is relatively easy for actors
to coordinate in their selection of policies to   
\vspace*{\parskip}
deploy in these settings.\\
\citet{cao2018emergent} look at negotiation between deep reinforcement learners, but not
between independently trained agents.
Several authors have recently investigated
the 
board game Diplomacy \citep{paquette2019no,NEURIPS2020_d1419302,gray2021humanlevel} which contains implicit bargaining problems amongst players. 
Bargaining
problems are also investigated in older MAL literature 
(e.g., \citealt{crandall2011learning}) as 
well as literature on automated negotiation
(e.g., \citealt{kraus2001strategic,baarslag2013evaluating}),
but also not between independently trained agents.
Considerable work has gone into understanding 
the emergence of norms in both human \citep{bendor2001evolution,boyd2009culture}
and artificial societies \citep{walker1995understanding,shoham1997emergence,sen2007emergence}. 
Especially relevant are empirical studies of bargaining across cultural contexts \citep{henrich2001search}. 
There is also recent multi-agent reinforcement learning
work on 
norms
\citep{hadfield2019legible,lerer2019learning,koster2020model} is also relevant here as bargaining problems can be understood as settings in which there are multiple efficient but incompatible norms.
However, much less attention has been paid in these literatures to how agents with different norms
are or aren't able to overcome normative disagreement. 

\par There are large
game-theoretic literatures on bargaining (for a review see \cite{muthoo2001economics}). 
This includes a long tradition of work on
\textit{cooperative
bargaining solutions}, which tries to 
establish normative principles for 
deciding among mutually beneficial 
agreements \citep{thomson1994cooperative}. 
We will draw on this work in our
discussion of normative (dis)agreement below.

\par Lastly, the class of norm-adaptive policies we develop in
Section \ref{sec:robustness-training} ---   
$\amTFT(\WelfareFunctions)$ --- 
can be seen as a more general variant of an approach proposed by 
\citet{boutilier1999sequential} for coordinating in pure coordination games. As
it implicitly searches for overlap in the agents' sets of allowed welfare functions,  
it is also similar to 
\citet{rosenschein1988deals}'s
approach to reaching agreement in general-sum games via sets of proposals by 
each agent.  

\section{Coordination, bargaining and normative disagreement}\label{sec:taxonomy} 
We are interested in a setting in which
multiple actors (``principals'') 
train machine learning systems offline, and
then deploy them into an environment 
in which they interact.
For instance, 
two different companies might 
separately train systems
to negotiate on their behalf and deploy them without explicit coordination 
on deployment. 
In this section, we develop a taxonomy of environments that these
agents might face, and relate these different types of
environments to the difficulty of bargaining.


\subsection{Preliminaries}

We formalize multi-agent environments as
partially 
observable stochastic games (POSGs). For simplicity
we assume two players, $i=1,2$. We will 
index player $i$'s counterpart by $-i$. 
Each player has an action space $\Act_i$. 
There is a space $\States$ of states
$S^t$ which 
evolve according to a Markovian transition function
$P(S^{t+1} \mid S^t, A_1^t, A_2^t)$. At
each time step, each player 
sees an observation $O_i^t$ which depends on 
$S^t$. Thus each player has an accumulating
history of observations $\Histories_i^t = 
\{ O_1^v, A_1^v \}_{v=1}^t$. We refer to the set
of all histories for player $i$ as 
$\Histories_i = \cup_{t=1}^\infty
\Histories^t_i$ and assume for simplicity that the initial 
observation history is fixed and common knowledge:
$h^0_1 = h^0_2 \equiv h^0$. 
Finally, principals choose among policies
$\pi_i: \Histories_i \rightarrow \Delta(\Act_i)$, which we imagine as artificial agents 
deployed by the principals.
We will refer to policy profiles generically 
as $\policyProfile \in \Pi:=\Pi_1 \times \Pi_2$.

\par Each player has a reward function $r_i$, 
such that $r_i(S^t, A_1^t, A_2^t)$ is their reward
at time $t$. We define the value to player $i$ of policy
profile $\policyProfile$ starting at history $h_i^t$ as
$V_i(h_i^t, \policyProfile) = 
\expect_\pi \left[ \sum_{v=t}^\infty \gamma^{v-t} r_i(
  S^v, A_1^v, A_2^v)
\mid H_i^t = h_i^t \right]$, where $\gamma \in  [0, 1)$ is a 
discount factor, and the value of a policy profile 
to player $i$ as 
$V_i(\policyProfile) = V_i(h^0, \policyProfile)$. A payoff profile 
is then a tuple $(V_1(\policyProfile), V_2(\policyProfile))$.
We say that $\pi$ is a (Nash) equilibrium of 
a POSG if $\pi_i \in \argmax_{\pi_i' \in \Pi_i}
V_i(h^0, \pi_i', \pi\mi)$ for $i=1,2$. We say that
$\pi$ is Pareto-optimal if for $i=1,2$ 
and $\pi' \in \Pi$ we have that
$V_i(\pi') > V_i(\pi)$ implies 
$V\mi(\pi') < V\mi(\pi).$ 

\subsection{Coordination problems}
We define a coordination problem as a game involving multiple Pareto-optimal equilibria 
(cf. \citealt{zhang2015equilibrium}), which require some coordinated action to achieve.
That is, if the players disagree about which equilibrium they are playing, they will not reach a Pareto-optimal outcome.
A \textit{pure coordination problem} is a game
in which there are multiple Pareto-optimal 
equilibria over which agents have identical interests. 
Although agents may still experience difficulties in pure coordination games, for instance due to a noisy communication channel, they are made easier by the fact that 
principals are indifferent between the Pareto-optimal equilibria.

\subsection{Bargaining problems and normative disagreement}\label{sec:bargaining}
We define a \textit{bargaining problem (BP)} to be a
game in which there are multiple  
Pareto-optimal equilibria over which the principals
have \textit{conflicting} preferences. 
These equilibria represent more than one way to collaborate for mutual benefit, or put in another way, for sharing a surplus.
Thus a bargaining problem is a mixed-motive coordination problem, in which there is conflicting interest between Pareto-optimal equilibria and common interest in reaching a Pareto-optimal equilibrium. 

\par We can distinguish between BPs which are \textit{symmetric} and
\textit{asymmetric} games. A 2-player game is symmetric if for any 
attainable payoff profile 
$(a, b)$ there exists a profile $(b, a)$.
The reason this distinction is important is
that all (finite) symmetric games have a symmetric Nash 
equilibrium \citep{nash1990non}; thus symmetric games
have a natural set of focal points \citep{schelling1958strategy} for aiding coordination in
mixed-motive contexts, while asymmetric BPs may not.
Similarly, given a chance to play a correlated 
equilibrium \citep{aumann1974subjectivity},
agents in a symmetric BP could play according to a 
correlated equilibrium which randomizes
using a symmetric distribution
over Pareto-optimal payoff profiles. 

Figure \ref{fig:coordination-problems} displays 
the payoff matrices of three coordination games: 
Pure Coordination, and two variants 
of Bach or Stravinsky (BoS), one of which is a 
symmetric BP
and one of which is an asymmetric BP.
Pure Coordination
is not a BP because it is not a mixed-motive game as the players only care about playing the same action.
On the other hand, in the case of symmetric BoS the players do have conflicting interest, however, there is a correlated equilibrium 
-- tossing a commonly observed fair coin -- 
that is
intuitively the most reasonable way of coordinating: It both maximizes the total payoff and offers each player the same expected reward.



\begin{figure}[H]
\small
\centering

\subfloat{
 \begin{game}{2}{2}
& B     & S\\
B   & $1,1$  & $0,0$\\
S   & $0,0$   & $1,1$
\end{game}
}
\qquad
\subfloat{
\begin{game}{2}{2}
      & B     & S\\
B   & $3,2$  & $0,0$\\
S   & $0,0$   & $2,3$
\end{game}
}
\qquad
\subfloat{
\begin{game}{2}{2}
      & B     & S\\
B   & $4,1$  & $0,0$\\
S   & $0,0$   & $2,2$
\end{game}
}
\caption{Payoff matrices for Pure Coordination (left), BoS (middle), Asymmetric BoS (right). 
}
\label{fig:coordination-problems}
\end{figure}

\par 
To 
develop a better intuition for
the sense in which equilibria can be more or less reasonable, consider a BoS with 
an
extreme asymmetry, with equilibrium payoffs (15, 10) and (1, 11). Even though each of these equilibria is Pareto-optimal, the latter 
seems unreasonable or \textit{uncooperative}: it yields a lower total payoff, more inequality, and lowers the reward of the worst-off player in the equilibrium. To formalize this intuition, we characterize the reasonableness of a Pareto-optimal payoff profile in terms of the extent to which it optimizes \textit{welfare functions}: we can say that (1, 11) is unreasonable because there is no (impartial, see below) welfare function that would prefer it.
Different welfare functions with different properties have been introduced in the literature (see Appendix \ref{sec:appendix_welfare_functions} for an overview), but two uncontroversial properties of a welfare function are
\emph{Pareto-optimality} (i.e., its optimizer 
should be Pareto-optimal) and 
\emph{impartiality}\footnote{The requirement of impartiality is also called \textit{symmetry} in the welfare 
economics literature, or \textit{anonymity} in the social choice literature  
\citep{campbell1980anonymity}.}
(i.e., the welfare of a policy profile
should be invariant to permutations of player 
indices).
From the latter property we can observe that the intuitively reasonable choice of playing the correlated equilibrium with a \textit{fair} correlation device in the case of symmetric games is also the choice which \textit{all} impartial
welfare functions recommend, provided that it is possible for the agents to play a correlated equilibrium.
\par
By contrast, in the asymmetric BoS from Figure \ref{fig:coordination-problems} we see that playing BB 
maximizes utilitarian welfare $\welfareUtil(\pi) = V_1(\pi) + 
V_2(\pi)$, whereas playing SS maximizes the egalitarian welfare $w^{\mathrm{Egal}}(\pi) = 
\min\{ V_1(\pi), V_2(\pi)\}$
subject to $\pi$ Pareto-optimal. 
Throwing a correlated fair coin to choose between the two would lead to an expected payoff that is 
optimal with respect to the Nash welfare $w^{\mathrm{Nash}}(\pi) = V_1(\pi) \cdot V_2(\pi)$. Each of 
these different equilibria has a normative principle to motivate it.

In the best case, all principals agree on the same welfare function as a common ground for coordination in asymmetric BPs. However, the principals may have reasonable differences with respect to which welfare function they perceive as fair,
and so they may train their systems to
optimize different welfare functions, leading to coordination failure
when the systems interact after deployment.
In cases where agents were independently-trained according to inconsistent welfare functions, we will say that there
is \textit{normative disagreement}.
There may be different
degrees of normative disagreement. For instance, the difference
$| \max_\pi \welfareUtil(\pi) - \max_\pi w^{\mathrm{Egal}}(\pi) |$
differs across games. 
\par

To summarize, we relate the difficulty of coordination problems to the concept of welfare functions: In pure coordination problems, they are not needed. In symmmetric bargaining problems, they all point to the same equilibria. And in asymmetric bargaining problems, they can serve to filter out intuitively unreasonable equilbria, but leave the possibility of normative disagreement between reasonable ones. This makes normative disagreement a critical challenge for bargaining.
In the remainder
of the paper, we will focus on asymmetric bargaining problems
for this reason.


\subsection{Norm-adaptive policies}

When there is potential for normative disagreement, 
it can be helpful for agents to have some flexibility as to the norms 
they play according to. 

\par A number of definitions of norms have been proposed in the social scientific literature (e.g., \citet{gibbs1965norms}), 
but they tend to agree that a norm is a rule specifying acceptable and unacceptable behaviors in a group of people, along with 
sanctions for violations of that rule. Game-theoretic work sometimes identifies norms with \textit{equilibrium selection devices} 
(e.g., \citealt{binmore1994economist,young1996economics}). Given that complex games generally exhibit many equilibria, 
some rule (such as maximizing a particular welfare function) is needed to select among them. 

\par Normative disagreement arises (among other reasons) from the 
underdetermination of good behavior in complex multi-agent settings. This is exemplified by the 
problem of conflicting equilibrium selection criteria in 
asymmetric bargaining problems, 
but there are other possible cases of undeterdetermination. One example is
undetermination of the \textit{beliefs} that a reasonable agent should act according to in games of incomplete 
information (cf. the common prior assumption in Bayesian games \citep{morris1995common}). Thus our definition of 
norm
will be more general than an equilibrium selection device, though in the remainder of the paper we will 
focus on the use of welfare functions to choose among equilibria. 

\begin{definition}[Norms] Given a 2-player POSG, a \textit{norm} $N$ is a tuple $N = (\Pi_1^N, \Pi_2^N, p_1, p_2, \Pi_1^p, \Pi_2^p)$, 
where $\Pi_i^N$ are normative policies (i.e., those which comply with the norm); $\Pi_i^p$ are ``punishment'' policies, which are enacted
when deviations from a normative policy are detected; and $p_i$ are rules for judging whether a deviation has happened and how to respond, i.e., 
$p_i: \Histories_i \rightarrow \{\texttt{True}, \texttt{False}\}$. A policy $\pi_i$ is compatible with norm $N$ if, for all $H_i \in \Histories_i$, 
\begin{equation*}
\pi_i(h) = \begin{cases}
                \pi_i^N(h), \pi^N_i \in \Pi_i^N, \text{ if } p_i(h) = \texttt{False}; \\
                \pi_i^p(h), \pi^p_i \in \Pi_i^p, \text{ if } p_i(h) = \texttt{True}.
             \end{cases}
\end{equation*}
\end{definition}
For example, in the iterated Asymmetric BoS, one norm is for both players to always play $B$ 
(this is the normative policy), 
and 
for a player to respond to deviations by continuing to play $B$. A similar norm 
is given by both players always playing $S$ instead. 
More generally, following the folk theorems \citep{mailath2006repeated}, 
an equilibrium in a repeated game corresponds to a norm, in which a 
profile of normative policies corresponds to play on the equilibrium path; 
the functions $p_i$ indicate whether there is deviation from equilibrium play; 
and punishment policies $\pi_i^p$ are chosen such that players are made
worse off by deviating than by continuing to play according to the normative
policy profile.

\par Now we formally define norm-adaptive policies.

\begin{definition}[Norm-adaptive policies.]
Take a 2-player POSG and a non-singleton 
set of norms $\Norms$ for that game. Let $N_i: \Histories_i \rightarrow \Norms$ be a surjective function that maps 
histories of observations to norms (i.e., for each norm in $\Norms$, there are histories for which 
$N_i$ chooses that norm.) 
Then, a policy $\pi_i'$ is \textit{norm-adaptive} if, for all $h$, there is a policy $\pi_i$
such that $\pi_i$ is compatible with $N_i(h)$ and $\pi'_i(h) = \pi_i(h)$.
\end{definition}
That is, norm-adaptive policies are able to play according to different norms depending on the circumstances.    
As we will see below, the 
benefit of making agents explicitly norm-adaptive is that this can help to prevent or resolve normative disagreement. 
Lastly, note that we can define higher-order norms and higher-order norm-adaptiveness: a higher-order norm 
is a norm $N$ such that policies $\Pi_i^N$ are themselves norm-adaptive with respect to some set of norms.
This framework allows us for discussing differing (higher-order) norms for resolving normative disagreement.  

\section{Multi-agent learning and cooperation failures in BPs}\label{sec:simple}

In this section, we 
illustrate how
cooperation-inducing, but norm-\textit{un}adaptive, multi-agent learning algorithms fail to cooperate in
asymmetric BPs.
In Section \ref{sec:robustness-training} we will then show how norm-adaptiveness improves cooperation.
The environments and algorithms considered are summarized in Table
\ref{tab:envs_and_algos}.


\subsection{Setup: Learning algorithms and environments}

In order to both include algorithms which use an explicitly specified welfare function and ones which do not, we use the
Learning with Opponent-Learning Awareness (LOLA) algorithm \citep{foerster2018learning} in its 
policy gradient and exact value function optimization versions as an example for the latter, and Generalized Approximate Markov Tit-for-tat ($\amTFT(w)$) for the former. 

We introduce $\amTFT(w)$ as a variant of
\citet{lerer2017maintaining}'s Approximate
Markov Tit-for-tat. 
The original 
algorithm trains a cooperative policy profile on the 
utilitarian welfare,
as well as a punishment policy, and switches
from the cooperative policy to the punishment policy
when it detects that the other player is
defecting from the cooperative policy. The algorithm has the appeal that it 
``cooperates with itself, is robust against defectors, and incentivizes cooperation from its partner''
\citep{lerer2017maintaining}. 
We consider the more general class of algorithms that we call $\texttt{amTFT}(w)$ in which a cooperative policy is 
constructed by optimizing an arbitrary welfare function $w$. Note that although 
$\texttt{amTFT}(w)$ takes a welfare 
function as an argument, the resulting policies are not norm-adaptive. 
To cover a range of environments representing both symmetric and asymmetric games,
we use some existing multi-agent reinforcement learning environments (IPD, 
Coin Game [CG; \citealt{lerer2017maintaining}]) and introduce two new ones (IAsymBoS, ABCG).


Iterated asymmetric Bach or Stravinsky (IAsymBoS) is an asymmetric version of the iterated Bach or 
Stravinsky matrix game. 
At each time step, the game defined on the right
in Figure \ref{fig:coordination-problems} is played. We focus
on the asymmetric variant due to the argument in Section
\ref{sec:bargaining} that players could resolve the symmetric
version by playing a symmetric equilibrium; however, applying
LOLA without modification would also lead to coordination failure
in the symmetric variant. It should also be noted that IAsymBoS is not an SSD because it does not incentivize defection: agents cannot gain from miscoordinating. Thus we consider IAsymBos to be a minimal example for an environment that can produce bargaining failures out of normative disagreement.

For a more involved example we also introduce an asymmetric version of the stochastic gridworld Coin Game \citep{lerer2017maintaining} -- 
asymmetric bargaining Coin Game (ABCG) -- which is both an SSD and an asymmetric BP.
In ABCG, a red and a blue agent navigate a grid with coins. Two coins simultaneously appear on this grid: a Cooperation coin and a Disagreement coin, each colored red or blue. 
Cooperation coins can only be consumed by both players moving onto them at the same time, whereas the Disagreement coin can only be consumed by the player of the same color as the coin.
The 
game is designed such that $\welfareUtil$ is maximized by both players always consuming the Cooperation coin, however, this will make one player benefit more than the other. Due to the sequential nature of the game, this means that welfare functions which also care about (in)equity will favor an equilibrium in which the player who benefits less from the Cooperation coin is allowed to consume the Disagreement coin from time to time without retaliation.

Players move simultaneously in all games. Note that we assume each player to have full knowledge of the other player's rewards as this is required by LOLA-exact during training and by $\amTFT(w)$ both at training and deployment.
We use simple tabular and neural network policy parameterizations, which are described in 
Appendix \ref{sec:experimental_details} along with learning algorithms.
\begin{table}[h]
\centering
\fontsize{8.5}{10.2} \selectfont
    \caption{Summary of the environments and learning algorithms that we use to 
    study sequential social dilemmas (SSDs) and asymmetric bargaining problems in this section. 
    }
    
    \smallskip
    \begin{tabular}{llll}  
\toprule
  \textbf{Environment} & \textbf{Asymmetric BP} & \textbf{SSD} & \textbf{Learning algorithms} \\\midrule
  Iterated Prisoner's Dilemma (IPD) & \xmark & \cmark &
  LOLA-Exact, 
  $\amTFT(w)$\\
  Iterated Asymmetric BoS (IAsymBoS) & \cmark & \xmark & LOLA-Exact, $\amTFT(w)$ \\
  Coin Game (CG) & \xmark & \cmark & $\text{LOLA-PG}$, $\amTFT(w)$ \\ 
  Asymmetric bargaining Coin Game (ABCG) & \cmark & \cmark & $\text{LOLA-PG}$, $\amTFT(w)$ 
  \\\bottomrule
    \end{tabular}
    \label{tab:envs_and_algos}
\end{table}

\subsection{Evaluating cooperative success}
\begin{figure}[h]
    \centering
    \includegraphics[width=130mm,keepaspectratio]{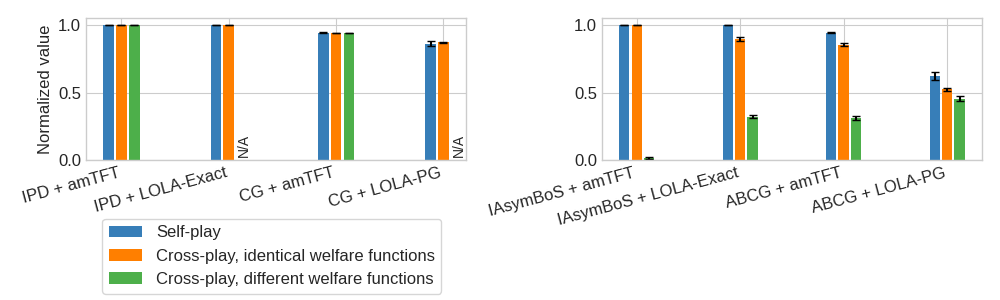}
    \caption{Comparison of cooperative success, measured by $\texttt{NormalizedScore}$, 
    between SSDs (left) and
    asymmetric BPs (right).
    Small black bars are standard errors. For LOLA in asymmetric BPs, payoff profiles are collected into sets
    corresponding to the welfare functions they \textit{implicitly} optimize, in order to obtain cross-play between identical and between different welfare functions.
    }
    \label{fig:results-all}
\end{figure}
\begin{wrapfigure}{R}{0.5\textwidth}
    \vspace{-10mm}

    \begin{center}
        \includegraphics[width=70mm,keepaspectratio]{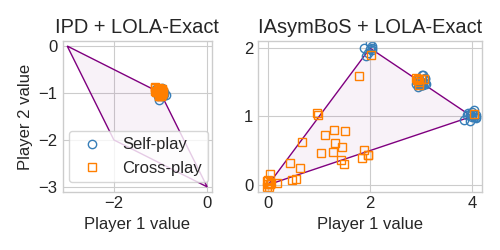}
    \end{center}
    \caption{Values of policy profiles trained in IPD and in IAsymBoS with LOLA in self-play 
    and cross-play. 
    In the IPD, all policy profiles cooperate in cross-play. 
    But in IAsymBoS, LOLA converges to different Pareto-optimal 
    outcomes in self-play, leading to failures in cross-play. 
    The purple areas describe sets of attainable payoffs. Jitter is added for improved visibility.}
    \label{fig:main-scatter-plot}
    \vspace{-5mm}
\end{wrapfigure}
We train policies in environments listed in Table \ref{tab:envs_and_algos} with the corresponding MARL algorithms.
After the training, we evaluate cooperative success in self-play and cross-play.
\textit{Self-play} 
refers to average performance between jointly-trained policies. 
\textit{Cross-play} refers to average performance between 
independently-trained policies.
We distinguish between two kinds of cross-play: 
that between agents trained using the same notion of collective optimality (such as a welfare function 
or an inequity aversion term in the value function), and between agents trained using different ones. 
For $\amTFT(w)$ we use two
welfare functions, $w^{\mathrm{Util}}$ and an inequity-averse welfare-function $w^{\mathrm{IA}}$
(see Appendix \ref{sec:experimental_details}). 

Comparing cooperative success across environments is not
straightforward: environments have different sets of feasible payoffs, 
and we should not evaluate cooperative success with a single welfare function,
because our concerns about normative disagreement stem from the fact that it is not obvious what single
welfare function to use.
Thus we compute cooperative success as follows.
First, we take $\WelfareFunctions$ to be the set of welfare functions which
we use in our experiments in the environment in
question. For instance, for IAsymBoS and $\amTFT(w)$ this is $\WelfareFunctions = \{ \welfareUtil, \welfareIneq\}$.
Second, we define disagreement payoff profiles $\pi^d$ corresponding to cooperation failure; in IAsymBoS the disagreement policy profile would be the one which has payoff profile $(0, 0)$.
Then, for a policy profile $\pi$, we compute its normalized score as 
$\mathrm{NormalizedScore}(\pi) \triangleq \max_{w \in \WelfareFunctions} \frac{w(\pi) - w(\pi^d)}{ \max_{\pi'} w(\pi') - w(\pi^d)}$. Under this scoring method, players do maximally well when they 
play a policy profile which is optimal according to some welfare function in $\WelfareFunctions$.


Figure \ref{fig:results-all} illustrates the difference between cooperation in SSDs and bargaining problems. When there is a single "cooperative" equilibrium, as in the case of IPD and CG, cooperation-inducing learning algorithms typically achieve cooperation in
cross-play.
In contrast, in IAsymBoS and ABCG we observe mild performance degradation in cross-play where agents optimize the same welfare functions, and strong degradation when agents optimize different welfare functions.

\section{Benefits and limitations of norm-adaptiveness}\label{sec:robustness-training}

The cooperation-inducing properties of the algorithms in Section \ref{sec:simple} 
are simple and are not
designed 
to help agents
resolve
potential normative disagreement
to avoid Pareto-dominated outcomes. The two main problems are (1) that the algorithms are ill-equipped for \textit{reacting} to normative disagreement, and (2) that they may confuse normative disagreement with defection.

The former problem is already evident in IAsymBoS. There, playing according to incompatible welfare functions is not interpreted as defection by $\amTFT$. This is not necessarily bad -- we claim that normative disagreement \textit{should} be treated differently to defection -- but it does mean that $\amTFT$ lacks the policy space to react to normative disagreement. For the latter problem, we observe that in the ABCG the $\amTFT(w)$-agent does classify some of the opponent's actions as defection, even though they are aimed to optimize an impartial welfare function, and punishes accordingly.

\subsection{$\amTFT(\WelfareFunctions)$}

\begin{wrapfigure}{R}{0.5\textwidth}
    \vspace{-5mm}
    \begin{center}
        \includegraphics[width=70mm,keepaspectratio]{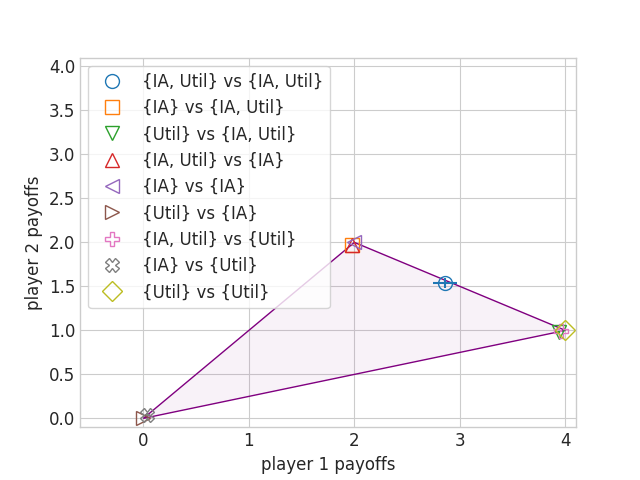}
    \end{center}
    \caption{Average cross-play payoff profiles with standard errors for profiles of policies
    $\left( \amTFT(\WelfareFunctions_1), \amTFT(\WelfareFunctions_2) \right)$. 
    The legend shows various combinations of sets $\WelfareFunctions_1$ and $\WelfareFunctions_2$. Coordination 
    failure only occurs when $\WelfareFunctions_1 \cap \WelfareFunctions_2 = \emptyset$.}
    \label{fig:amtftW}
    \vspace{-5mm}
\end{wrapfigure}

To overcome both of these problems, we propose a 
norm-adaptive modification to $\amTFT(w)$.
As we only aim to illustrate the benefit of norm-adaptive policies, we keep the implementation simple: Instead of a 
welfare 
function $w$ the algorithm now takes a welfare-function set $\WelfareFunctions$: $\amTFT(\WelfareFunctions)$. The agent starts out 
playing according to some $w \in \WelfareFunctions$. The initial choice can be made at random or according to a preference ordering 
between the $w \in \WelfareFunctions$. 

\par $\amTFT(\WelfareFunctions)$ then follows a two-stage decision process. 
First, if the agent detects that their opponent is merely playing according to a different welfare  function than itself, $w$ gets re-sampled. 
Second, if the agent detects defection by the opponent, it will punish, just as in the original algorithm. However, 
by first checking for normative disagreement, we make sure that punishment does not happen as a result of a 
normative disagreement. 

\par In Figure \ref{fig:amtftW} we illustrate how, assuming uniform re-sampling 
from $\WelfareFunctions$, 
$\amTFT(\WelfareFunctions)$ can 
overcome normative disagreement and perform close to the Pareto frontier. 
Agents need not
sample uniformly, though. 
For instance when the number of possible welfare functions is large, it would be beneficial to 
put higher probability on welfare functions which one's counterpart is more likely to be optimizing. Furthermore an agent might want to put higher probability on welfare functions that it prefers.

\par Notice that, when $w \in \WelfareFunctions$, one player using $\amTFT(\WelfareFunctions)$ 
rather than $\amTFT(w)$ leads to a (weak) Pareto improvement. 
Beyond that, in 
anticipation of bargaining failure due to not having a way to resolve normative disagreement, players are 
incentivized to include more than just their preferred welfare function into $\WelfareFunctions$. 
In both IAsymBoS (see Figure \ref{fig:amtftW}) and ABCG (see Table \ref{tab:payoffs-abcg}, Appendix 
\ref{sec:more-results}) we can observe significant improvement for cross-play when at least one player is 
norm-adaptive.

\subsection{The exploitability-robustness tradeoff}
As our experiments with $\amTFT(\WelfareFunctions)$ show, agents 
who are more flexible are less prone to bargaining
failure due to normative disagreement. However, they are prone
to having that flexibility exploited by their
counterparts. For instance, an agent which is open to optimizing either $\welfareUtil$ or $\welfareIneq$ will end 
up optimizing $\welfareIneq$ if playing against an agent for whom $\WelfareFunctions = \{\welfareIneq\}$. 
More generally,
an agent who puts higher probability on a welfare function it has a preference for, when 
re-sampling $w$, will be less robust against counterparts who disprefer that welfare function and put a lower 
probability on it. An agent who tries to guess the counterpart's welfare function and tries to accommodate to this 
is exploitable to agents who do not.

\section{Discussion}
We formally introduce a class of hard cooperation problems --
asymmetric 
bargaining problems 
-- and situate them within a wider game taxonomy.
We 
argue that they are hard because there can arise 
normative disagreement between multiple "reasonable" cooperative equilibria, 
characterized by divergence in the preferred outcomes according to different welfare
functions.
This presents a problem for those deploying AI systems without coordinating on 
the norms those systems follow. 
We have introduced the notion of \textit{norm-adaptive} policies, which are
policies that allow agents to change the norms according to which they play, 
giving rise to opportunities for resolving normative disagreement.  
As an example of a class of norm-adaptive policies, we introduced 
$\amTFT(\WelfareFunctions)$, and showed in experiments that this tends to 
improve robustness to normative disagreement. 
On the other hand, we have demonstrated a 
\textit{robustness-exploitability tradeoff}, under which methods
that learn more normatively flexible strategies 
are also more vulnerable to exploitation.

\par There are a number of limitations to this work.
We have throughout assumed that the agents have a 
common and correctly-specified model of their environment, 
including their counterpart's reward function. In real-world
situations, however, principals may not have identical 
simulators with which to train their systems, and there
are well-known obstacles to the honest disclosure of 
preferences
\citep{hurwicz1972informationally}, 
meaning that common knowledge of rewards
may be unrealistic.
Similarly, we assumed a certain degree of reasonableness on part of the 
principals, see by the willingness to play the symmetric correlated 
equilibrium in symmetric BoS (Section 
\ref{sec:taxonomy}), for instance. However, we believe this to be a minimal 
assumption as the deployers of such agents are aware of the risk of 
coordination failure as a result of insisting on equilibria that no impartial welfare function would recommend.

\par Future work should consider more sophisticated and learning-driven 
approaches to designing 
norm-adaptive policies, as 
$\amTFT(\WelfareFunctions)$ relies on a finite set of user-specified welfare functions 
and a hard-coded procedure for switching between policies. 
One possibility is to train 
agents who are themselves capable of 
jointly deliberating
about the principles they should use to select an 
equilibrium, e.g., 
deciding among
the axioms which characterize
different bargaining solutions (see Appendix \ref{sec:appendix_welfare_functions})
in the hopes that they will be able to resolve
initial disagreements. Another direction is resolving disagreements
that cannot be expressed as disagreements over the welfare functions 
according to which agents play; for instance, disagreements over the
beliefs or world-models which should inform agents' behavior. 

\begin{ack}
We'd like to thank Yoram Bachrach, Lewis Hammond, Vojta Kovařík, Alex Cloud, as well as our anonymous reviewers for their valuable feedback; Daniel Rüthemann for designing Figures \ref{fig:coin-game-diagram} and \ref{fig:ab-coin-game-diagram}; Chi Nguyen for crucial support just before a deadline; Toby Ord and Jakob Foerster for helpful comments.

Julian Stastny performed part of the research for this paper while interning at the Center on Long-Term Risk. Johannes Treutlein was supported by the Center on Long-Term Risk, the Berkeley Existential Risk Initiative, and Open Philantropy. Allan Dafoe received funding from Open Philantropy and the Centre for the Governance of AI.

\end{ack}

\bibliographystyle{plainnat}
\bibliography{refs}

\clearpage
\appendix
\setcounter{page}{1}
\section{Welfare functions}
\label{sec:appendix_welfare_functions}

Different welfare functions have been introduced in the literature. Table \ref{tab:welfare} gives an overview over commonly discussed welfare functions. Their properties are noted in Table \ref{tab:welfare2}. 

\begin{table}[H]
    \caption{
    Welfare functions, adapted to the multi-agent RL 
    setting where two agents with value functions $V_1, V_2$
    are bargaining over the policy profile
    $\policyProfile$ 
    to deploy. 
    $d_i$, the \textit{disagreement value}, is the value which player i gets when bargaining fails.
    } 
    \smallskip
    \centering
    \bgroup
    \def\arraystretch{1.5}
    \begin{tabular}{cc}
    \toprule
         \textbf{Name of welfare function} $w$ & \textbf{Form of} 
         $w(\pi)$ \\ \midrule
         Nash \citep{nash1950bargaining} &  
            $            \left[ V_1(\policyProfile) - d_1 \right] 
            \cdot \left[ V_2(\policyProfile) - d_2 \right]$ \\ 
        \makecell{Kalai-Smorodinsky \\ \citep{kalai1975other}} & 
            \makecell{$
            \Big\lvert 
            \frac{V_1(\policyProfile) - d_1}{V_2(\policyProfile) - d_2} -
            \frac{\sup_{\policyProfile} V_1(\policyProfile) - 
            d_1}{\sup_{\policyProfile}V_2(\policyProfile) - d_2}
            \Big\rvert$ s.t. Pareto-optimal} \\ 
        Egalitarian \citep{kalai1977proportional} & 
          $\min\left\{ 
          V_1(\policyProfile) - d_1,
          V_2(\policyProfile) - d_2
          \right\}$ s.t. Pareto-optimal \\ 
        Utilitarian  \citep{harsanyi1955cardinal}
        & $V_1(\policyProfile) + V_2(\policyProfile)$ \\ \bottomrule 
    \end{tabular}
    \vspace{10pt}
    \egroup
    \label{tab:welfare}
\end{table}

\par
\textit{Pareto-optimality} refers to the property that the welfare function's optimizer should be Pareto-optimal. \textit{Impartiality}, often called symmetry, implies that the welfare of a policy profile should be invariant to permutations of player indices. These are treated as relatively uncontroversial properties in the literature.
\textit{Invariance to affine transformations} of the payoff matrix is usually motivated by the assumption that interpersonal comparison of utility is impossible. In contrast, the utilitarian welfare function assumes that such comparisons are possible.
\textit{Independence of irrelevant alternatives} refers to the principle that a preference for an equilibrium A over equilibrium B should only depend on properties of A and B. That is, a third equilibrium C should not change the preference ordering between A and B.
\textit{Rescource monotonicity} refers to the principle that if the payoff for any policy profile increases, this should not make any agent worse off.


\begin{table}[H]
\centering
\caption{Properties of welfare functions. For properties of Nash welfare here the set of feasible payoffs is assumed to be convex (\citet{zhou1997nash} describes properties in the non-convex case).}
\smallskip
\fontsize{8.5}{10.2} \selectfont
\begin{tabular}{ccccc}
\toprule
 & Nash & Kalai-Smorodinsky & Egalitarian & Utilitarian
\\\midrule
Pareto optimality & \cmark &  \cmark &  \cmark &  \cmark 
\\
Impartiality & \cmark &  \cmark &  \cmark &  \cmark 
\\
Invariance to affine transformations &  \cmark & \cmark & \xmark & \xmark 
\\
Independence of irrelevant alternatives & \cmark &   \xmark &  \cmark &  \cmark 
\\
Resource monotonicity & \xmark &   \cmark &  \cmark & \xmark 
\\\bottomrule
\end{tabular}
\label{tab:welfare2}
\end{table}

    

\section{Environments}

\subsection{Additional descriptions of environments}

\subsubsection*{Coin Game}

\begin{figure}[H]
    \centering
    \includegraphics[height=70mm, keepaspectratio]{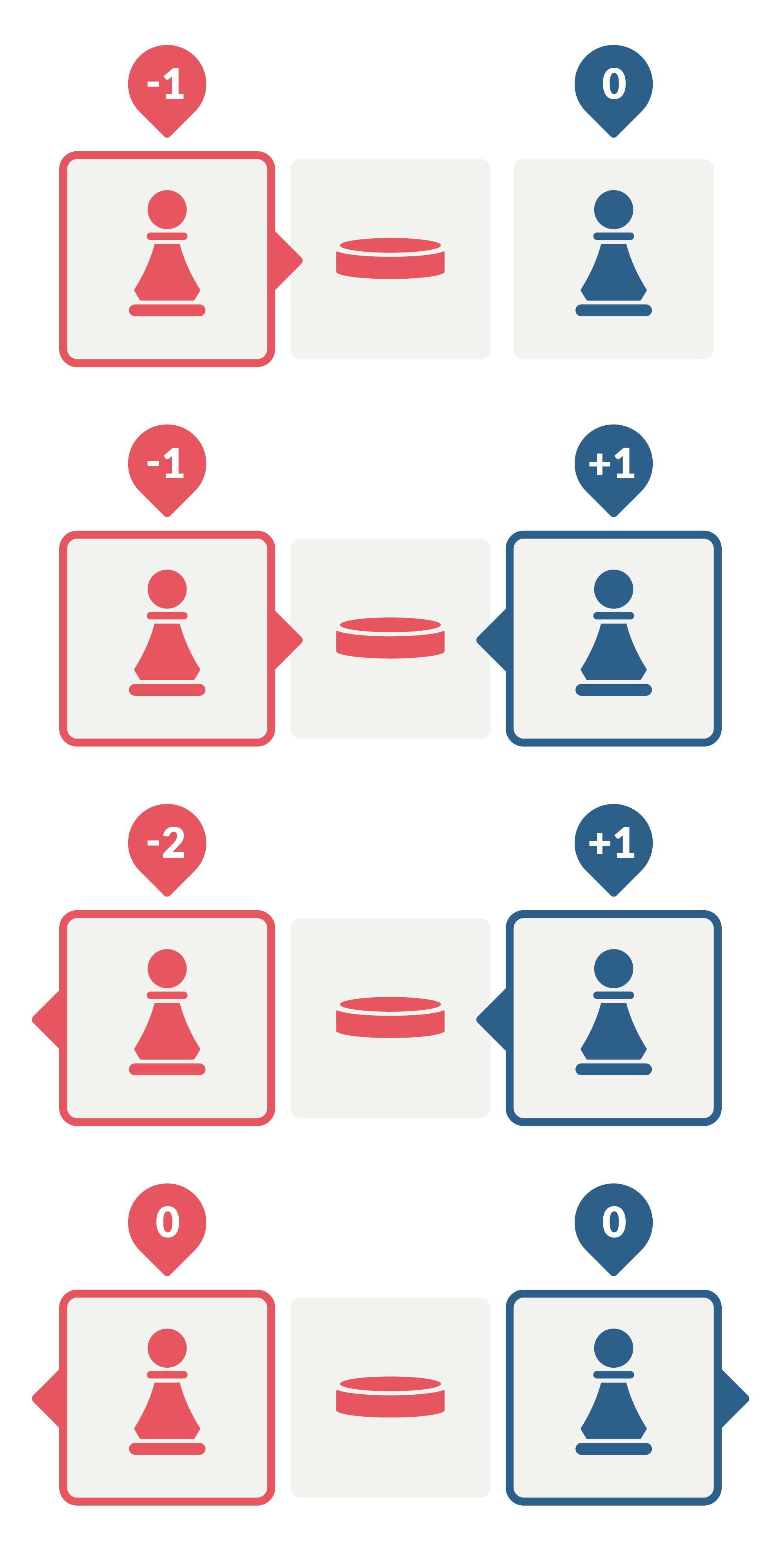}
    \caption{A pictorial description of agents' rewards in the Coin Game environment \citep{lerer2017maintaining}. A Red player and a Blue player
navigate a grid and pick up randomly-generated coins. Each player gets a reward of $1$
for picking up a coin of any color. But, a player gets a reward of $-2$
if the other 
player picks up their coin. This creates a social dilemma in which the 
socially optimal behavior is to only get one's own coin, but there is 
incentive to defect and try to get the other player's coin as well.}
    \label{fig:coin-game-diagram}
\end{figure}

\subsubsection*{Asymmetric bargaining Coin Game}

\begin{figure}[H]
    \centering
    \includegraphics[width=130mm, keepaspectratio]{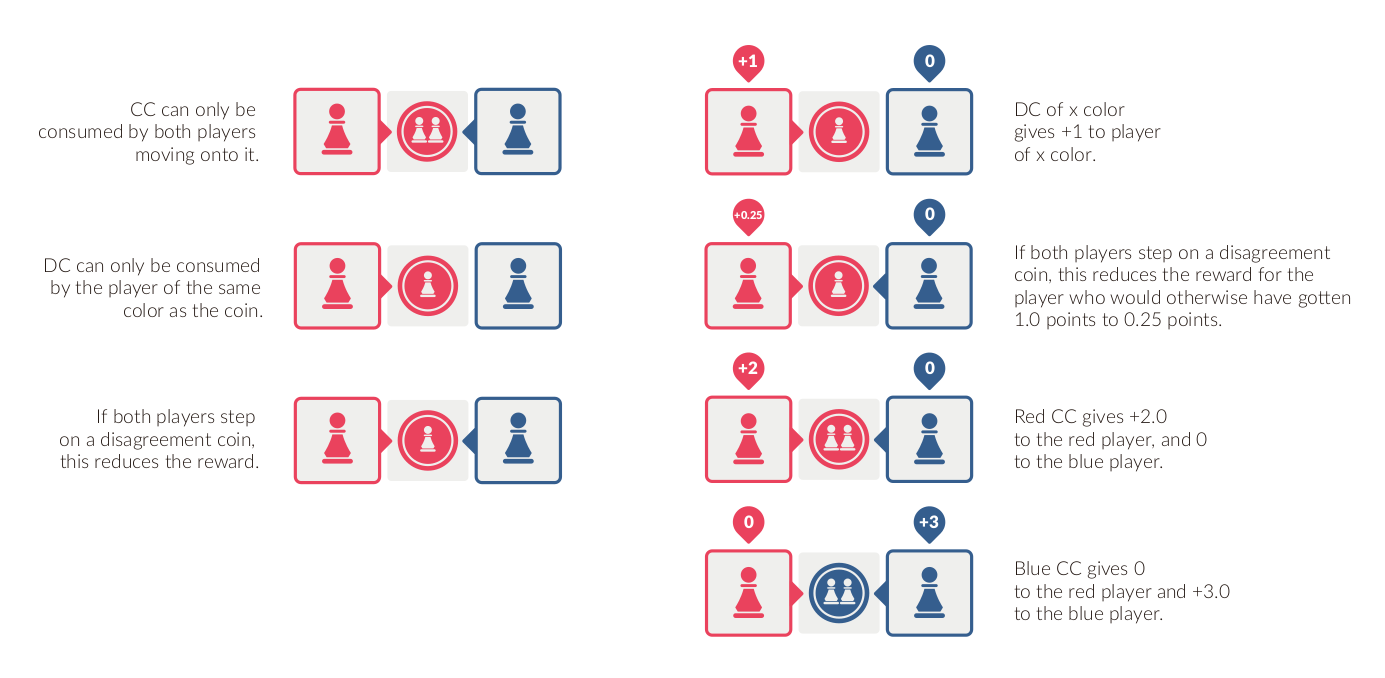}
    \caption{Illustration of the Asymmetric bargaining Coin Game described in Section \ref{sec:simple}. CC stands for Cooperative Coin and DC stands for Disagreement Coin.}
    \label{fig:ab-coin-game-diagram}
\end{figure}

\subsection{Disagreement profiles for each environment}

For computing the normalized scores in Figure \ref{fig:results-all}, we use the following disagreement profiles, corresponding to cooperation failure.
In IPD, the cooperation failures are only the joint defections (D,D) producing a reward profile of (-3, -3).
In IAsymBoS, they are the profiles (B, S) and (S, B), both associated with the reward profile (0, 0).
In Coin Game, the cooperation failures are when both players pick all coins at maximum speed selfishly and this produce a reward profile of (0, 0).
In Asymmetric bargaining Coin Game (ABCG), when the players fails to cooperates, they can punish the opponent by preventing any coin to be picked by waiting instead of picking their own defection coin (e.g. with $\amTFT$). In these cooperation failures, the reward profile with perfect punishment is (0, 0).
\section{Experimental details}
\label{sec:experimental_details}

\subsection{Learning algorithms}

We use the discount factor $\gamma=0.96$ for all the algorithms and environments unless specified 
otherwise in Table \ref{table:hyperparameters}.

\subsubsection*{Approximate Markov tit-for-tat ($\amTFT$)}

We follow the $\amTFT$ algorithm from \citep{lerer2017maintaining} with two changes:
1) Instead of using a selfish policy to punish the opponent, we use a policy that minimizes the reward of the opponent. 
2) We use the $\amTFT$ version that uses rollouts to compute the debit and the punishment length. We observed that using long rollouts to compute the debit increases significantly the variance on the debit and this leads to false positive detection of defection. To reduce this variance, we thus compute the debit without rollouts. We use the direct difference between the rewards given by the actual action and given by the simulated cooperative action. The rollout length used to compute the punishment length is 20. 

We note that in asymmetric BoS, training runs using the utilitarian welfare function sometimes learn to produce, in self-play, the egalitarian outcome instead of the utilitarian outcome. In these cases the policy gets discarded.

The inequity-averse welfare $\welfareIneq$ is defined as follows:

\begin{equation*}\label{eq:inequity-aversion}
    \begin{aligned}
        \welfareIneq(\pi) =
          V_1(\pi) + V_2(\pi) - \beta\cdot \expect 
          \left[ \lvert e^t_1(h^t_1, h^t_2) - e^t_2(h^t_1, 
          h^t_2) \rvert \right]
    \end{aligned}
\end{equation*}
%
where
$e_i^t(h^t_1, h^t_2) = \gamma \lambda e^{t-1}_i(h^{t-1}_1, 
h^{t-1}_2) + r_i(S^t, A_1^t, A_2^t)$
are smoothed cumulative rewards with
a discount factor $\gamma$ and a parameter $\beta$ which controls how much unequal outcomes are penalized.

\subsubsection*{Learning with Opponent-Learning Awareness (LOLA)}
Write $V_i(\theta_1, \theta_2)$ as the value to player
$i$ under a profile of policies with parameters
$\theta_1, \theta_2$. Then, the LOLA update \citep{foerster2018learning}
for player 1 at time $t$ with 
parameters $\delta, \eta > 0$ is
\begin{equation*}\label{eq:lola}
\theta^t_1 = \theta^{t-1}_1
+ \delta \nabla_{\theta_1} V_i(\theta_1, \theta_2) +  
\delta \eta
\left[ \nabla_{\theta_2} V_1(\theta_1, \theta_2) \right]^{\Tra}
\nabla_{\theta_1} \nabla_{\theta_2} V_2(\theta_1, \theta_2).
\end{equation*}
%


\subsection{Policies}

We used RLlib \citep{pmlr-v80-liang18b} for almost all the experiments. It is under Apache License 2.0. All activation functions are ReLU if not specified otherwise.

\subsubsection*{Matrix games (IPD and IBoS) with LOLA-Exact}
We used the official implementation of LOLA-Exact from https://github.com/alshedivat/lola. We 
slightly modified it to increase the stability and remove a few confusing behaviors. 
Following \citet{foerster2018learning}'s parameterization of policies in the iterated 
Prisoner's Dilemma, we use policies which condition on the previous pair of actions played, 
with the difference that instead of using one parameter for every possible previous action 
profile, we use $N$ parameters for every one of the previous action profile to always to play 
in larger action space ($N$ actions possible).
\subsubsection*{Matrix games (IPD and IBoS) with amTFT}
We use a Double Dueling DQN architecture + LSTM for both simple and complex environments when using amTFT. This architecture is non-exhaustively composed of a shared fully connected layer (hidden layer size 64), an LSTM (cell size 16), a value branch and an action-value branch both composed of a fully connected network (hidden layer size 64).  
\subsubsection*{Coin games (CG and ABCG) with LOLA-PG}
We used the official implementation from https://github.com/alshedivat/lola, which is under MIT license. We slightly modified it to increase the stability and remove a few confusing behaviors. Especially, we removed the critic branch, which had no effect in practice. We use a PG+LSTM architecture composed of two convolution layers (kernel size 3x3 and feature size 20), an LSTM (cell size 64) and a final fully connected layer.


ABCG + $\text{LOLA-PG}$ mainly generates policies that are associated with an egalitarian welfare function. Within this set of policies, some are a bit closer to the utilitarian outcome than others, which we used as a basis to classify them as "utilitarian" for the purpose of comparison in Figures \ref{fig:results-all} and \ref{fig:all_scatter_plots}. However, because the difference is small, we do not observe a lot of additional cross-play failure compared to cross-play between the "same" welfare function.
It should also be noted that we chose to discard the runs in which none of the agents becomes competent at picking any of the coins. 

\subsubsection*{Coin games (CG and ABCG) with amTFT}
We use a Double Dueling DQN architecture + LSTM for both simple and complex environments when using amTFT. The architecture used is non-exhaustively composed of two convolution layers (kernel size 3x3 and feature size 64), an LSTM (cell size 16), a value branch and an action-value branch both composed of a fully connected network (hidden layer size 32).  


\subsection{Code assets}

The code that we provide allows to run all of the experiments and to generate the figures with our results.
All instructions on how to install and run the experiments are given in the `README.md` file.
The code to run the experiments from Section \ref{sec:simple} is in the folder "base\_game\_experiments". 
The code to run the experiments from Section \ref{sec:robustness-training}, is in the folder `base\_game\_experiments`. 
The code to generate the figures is in the folder `plots`.

An anonymized version of the code is available at \href{https://github.com/68545324/anonymous}{https://github.com/68545324/anonymous}.

\subsection{Hyperparameters}

\begin {table}[H]
\centering
\caption {Main hyperparameters for each cross-play experiment.}
\smallskip
\resizebox{\columnwidth}{!}{%
\begin{tabular}{cccccc}
\textbf{Env. (Algo.)}                     & \textbf{LR}       & \textbf{Episodes} & \textbf{Episode length}    & \textbf{Gamma} & \textbf{LSTM}     \\ \toprule
Matrix games (LOLA-Exact)                   & 1.0               & N/A               & N/A               & 0.96  & N/A           \\ 
Matrix games (amTFT)                        & 0.03              & 800               & 20                & 0.96  & \checkmark             \\ 
CG (LOLA-PG)                               & 0.05              & 2048000           & 40                & 0.90  & \checkmark             \\ 
ABCG (LOLA-PG)                               & 0.1              & 2048000           & 40                & 0.90  & \checkmark             \\ 
CG and ABCG (amTFT)                                 & 0.1               & 4000              & 100               & 0.96  & \checkmark             \\ 
\bottomrule
\end{tabular}
    \label{table:hyperparameters}%
}
\end {table}

All hyperparameters can be found in the code provided as the supplementary material. They are stored in the "params.json" files associated with each replicate of each experiment. The experiments are stored in the "results" folders.

The hyperparameters selected are those producing the most welfare-optimal results when evaluating in self-play (the closest to the optimal profiles associated with each welfare function). Both manual hyperparameter tuning and grid searches were performed.

\section{Additional experimental results}
\label{sec:more-results}
\begin{table}[H]
\caption{Cross-play normalized score of $\amTFT(\WelfareFunctions)$ in ABCG.}
\scriptsize
\centering
\begin{tabular}{cc|c|c|c|}
& \multicolumn{1}{c}{} & \multicolumn{3}{c}{}\\
& \multicolumn{1}{c}{} & \multicolumn{1}{c}{$\{ \welfareUtil \}$}  
& \multicolumn{1}{c}{$\{\welfareIneq\}$} 
&
\multicolumn{1}{c}{$\{\welfareUtil, \welfareIneq\}$} 
  \\\cline{3-5}
\multirow{3}*{}  & $\{ \welfareUtil \}$ &  0.82 & 
    0.12 & 0.69
\\\cline{3-5}
    & $\{ \welfareIneq \}$ 
    & 0.92 & 0.95 & 0.95  \\ \cline{3-5}
    & $\{ \welfareUtil, \welfareIneq \}$ 
    & 0.85 & 0.91 & 0.91  \\ \cline{3-5}
\end{tabular}
\label{tab:payoffs-abcg}
\end{table}

\begin{figure}[H]
    \centering
    %
    \includegraphics[width=130mm,keepaspectratio]{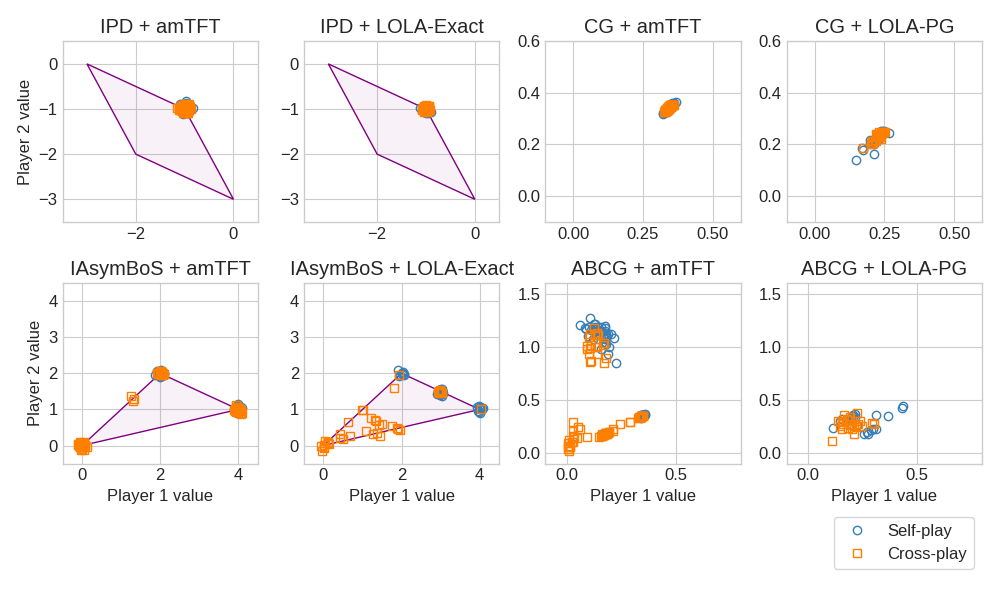}
    \includegraphics[width=130mm,keepaspectratio]{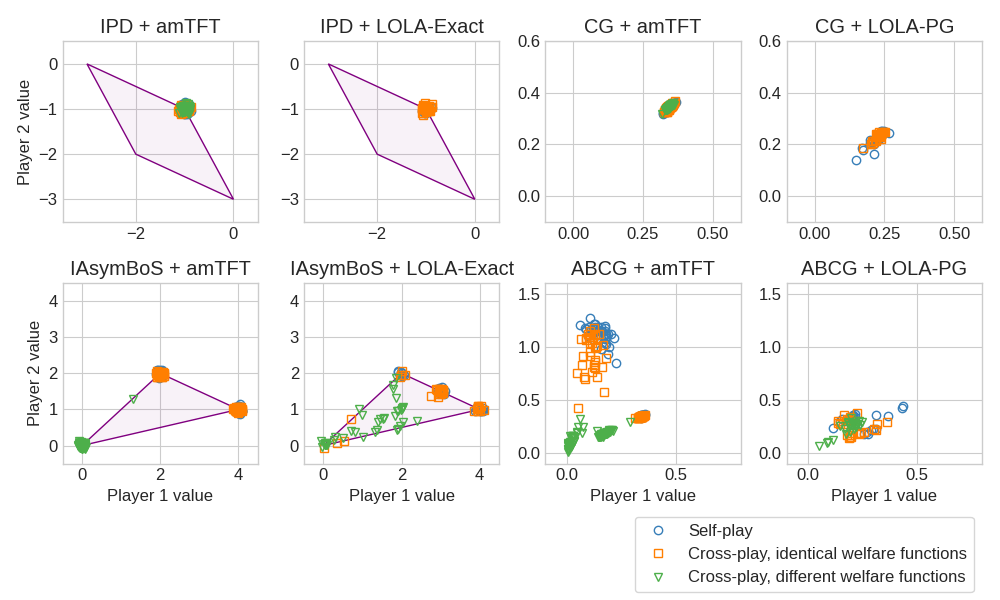}
        \caption{Mean reward of policy profiles for environments and learning algorithms given in Table \ref{tab:envs_and_algos}
        The purple areas describe sets of attainable payoffs. In matrix games, a small amount of jitter is added to the points for improved visibility. The plots on top compare self-play with cross-play, whereas the plots below 
        compare
        cross-play between policies optimizing same and different welfare functions. 
        }
    \label{fig:all_scatter_plots}
\end{figure}

\section{Exploitability-robustness tradeoff}

In this section, we provide some theory on the tradeoff between exploitability and robustness.
Consider some asymmetric BP, with welfare functions \(w,w'
\) optimized in equilibrium, such that for any two policy profiles \(\policyProfile^w,\policyProfile^{w'}\), the cross-play policy profile \((\policyProfile^w_1,\policyProfile^{w'}_2)\) is Pareto-dominated. Note that the definition implies that the welfare functions must be distinct.

We can derive an upper bound for how well players can do under the cross-play policy profile. It follows from the fact that both \(\policyProfile^w\) and \(\policyProfile^{w'}\) are in equilibrium that for \(i=1,2\), it is
\begin{equation}V_i(\policyProfile^w_1,\policyProfile^{w'}_2)\leq \min \{V_i(\policyProfile^w),V_i(\policyProfile^{w'})\}.\end{equation}
This is because, otherwise, at least one of the two profiles cannot be in equilibrium, since a player would have an incentive to switch to another policy to increase their value. 

From the above, it also follows that the cross-play policy must be strictly dominated. To see this, assume it was not dominated. This would imply that one player has equal values under both profiles. So that player would be indifferent, while one of the profiles would leave the other player worse off. Thus, that profile would be weakly Pareto dominated, which is excluded by the definition of a welfare function.

It is a desirable quality for a policy profile maximizing a welfare function in equilibrium to have values that are close to this upper bound in cross-play against other policies. For instance, if some coordination mechanism exists for agreeing on a common policy to employ, it may be feasible to realize this bound against any opponent willing to do the same. 

Moreover, the bound implies that whenever we try to be even more robust against players employing policies corresponding to other welfare functions (e.g., a policy which reaches Pareto optimal outcomes against a range of different policies), our policy will cease to be in equilibrium. In that sense, such a policy will be exploitable, while unexploitable policies can only be robust against different opponents up to the above bound. Note that this holds even in idealized settings, where coordination, e.g., via some grounded messages, is possible.

Lastly, note that if no coordination is possible, or if no special care is being taken in making policies robust, then equilibrium profiles that maximize a welfare function can perform much worse in cross-play than the above upper bound.

We show this formally in a special case in which our POSG is an iterated game, i.e., it only 
has one state. Moreover, we assume that it is an asymmetric BP, and that for both welfare 
functions \(w,w'\) in question, an optimal policy \(\pi^w,\pi^{w'}\) exists that is 
deterministic. Denote \( V_i(t,\pi) \triangleq V_i(H_t,\pi) \)  
(where \(H_t\) denotes the joint observation history up 
to step \(t\)) as the return from time step \(t\) onwards, under the policy \(\pi\). We also 
assume that for any player \(i\), the minimax value
\(\min_{a_{-i}}\max_{a_i}\frac{1}{1-\gamma}r_i(a)\)
is strictly worse than the values of their preferred welfare function maximizing profile. Then we can show that policies maximizing the welfare functions exist such that after some time step, their returns will be upper bounded by their minimax return.

\begin{prop} In the situation outlined above, there exist policices \(\tilde{\pi}^w,\tilde{\pi}^{w'}\) optimizing the respective welfare functions and a time step \(t\) such that
\[V_i(t, \tilde{\pi}_1^w,\tilde{\pi}_2^{w'})\leq \min_{a_{-i}}\max_{a_i}\frac{1}{1-\gamma}r_i(a)\]
for \(i=1,2\).
\end{prop}
\begin{proof}
Define \(\tilde{\pi}^w\) as the policy profile in which player \(i\) follows \(\pi_i^w\), unless the other player's actions differ from \(\pi_{-i}^{w}\) at least once, after which they switch to the action
\(a_{i} \triangleq \argmin_{a_{i}}\max_{a_{-i}}r_{-i}(a)\).
Define \(\tilde{\pi}^{w'}\) analogously. Note that both profiles are still optimal for the corresponding welfare functions.

As argued above, the cross-play profile \((\pi_1^w,\pi_2^{w'})\) must be strictly worse than their preferred profile, for both players. So there is a time \(t'\) after which an action of a player \(i\) must differ from \(-i\)'s prefered profile and thus \(-i\) switches to the minimax action \(a_{-i}\). We have \(V_i(t',\pi_i^w,a_{-i})\leq \min_{a'_{-i}}\max_{a'_i}\frac{1}{1-\gamma}r_i(a')\), i.e., the value \(i\) gets after step \(t'\) must be smaller than their minimax value, and by assumption, the minimax value is worse for \(i\) than the value of their preferred welfare-maximizing profile. Hence, there must be a time step \(t-1\geq t'\) after which \(i\) also switches to their minimax action.

From \(t\) onwards, both players play \(a\), so
\[V_i(t,\tilde{\pi})
= \frac{1}{1-\gamma}r_i(a)\leq \max_{a'_i}\frac{1}{1-\gamma}r_i(a_{-i},a'_i)=
\min_{a'_{-i}}\max_{a'_i}\frac{1}{1-\gamma}r(a').\]
\end{proof}

\end{document}